\documentclass[journal,onecolumn,12pt]{IEEEtran}
\usepackage{amsmath, amssymb, scalerel, amsthm, cite, xcolor, subfig}

\newtheorem{theorem}{Theorem}

\newtheorem{definition}{Definition}

\usepackage{wrapfig}
\usepackage{algorithm}
\usepackage[noend]{algpseudocode}
\usepackage{caption}
\usepackage[left=1in,top=1in,right=1in,bottom=1in]{geometry}
\title{Optimal Placement of UAVs for Minimum Outage Probability}
\author{Maryam Shabanighazikelayeh and Erdem Koyuncu\thanks {The authors are with the Department of Electrical and Computer Engineering, University of Illinois at Chicago. Email: \{mshaba7, ekoyuncu\}@uic.edu.
    } \thanks{Part of this work \cite{koyuncupimrc19} was presented at the IEEE International Symposium on Personal, Indoor and Mobile Radio Communications in September 2019.} \thanks{This work was supported in part by the NSF Award CCF--1814717.}}
    
\begin{document}
\maketitle
\vspace{-20pt}
\begin{abstract}
We consider multiple unmanned aerial vehicles (UAVs) serving a density of ground terminals (GTs) as base stations. The objective is to minimize the outage probability of GT-to-UAV transmissions. Optimal placement of UAVs under different UAV altitude constraints and GT densities is studied. First, using a random deployment argument, a general upper bound on the optimal outage probability is found for any density of GTs and any number of UAVs. A matching lower bound is also derived to show that the optimal outage probability decays exponentially with the number of UAVs. Next, the structure of optimal deployments is studied when the common altitude constraint is large. For a wide class of GT densities, it is shown that all UAVs should be placed to the same location in an optimal deployment.  A design implication is that one can use a single multi-antenna UAV as opposed to multiple single-antenna UAVs without loss of optimality. This result is also extended to a practical variant of the Rician fading model recently developed by Azari et al. for UAV communications. 
Numerical deployment of UAVs in the centralized and practical distributed settings are carried out using the particle swarm optimization and modified gradient descent algorithms, respectively.
\end{abstract}

\begin{IEEEkeywords}
UAV-aided communications, optimal placement, outage probability, distributed algorithms.
\end{IEEEkeywords}
\section{Introduction}
\label{secIntro}
\subsection{Background and Motivation}
Unmanned aerial vehicles (UAVs) have recently been successively utilized in many diverse areas, including wireless communications \cite{zeng2016wireless}. Several applications have been mentioned for UAV-assisted communications including UAVs acting as base stations \cite{bor2016efficient, lyu1}, and relays \cite{uavrelay1, zhan2011wireless, zeng2016throughput, 8580994}. The ability of UAVs acting as data collection units \cite{pearre1,cyclical} makes them particularly appealing for Internet of Things applications \cite{zeng2016wireless,myIoT}. In fact, UAVs' anywhere, anytime relocation ability plays a significant role in their success in improving the performance of various wireless systems. In this context, UAVs can be utilized to improve network power-efficiency \cite{ouruav1, ouruav2}.

Despite their effectiveness in wireless communications, there are still several fundamental challenges in UAV-based systems that are yet to be resolved. For example, trajectory optimization and optimal placement of UAVs is an important problem in designing UAV-aided wireless communication systems. The problem of optimal placement for UAVs, even in a scenario where the location of users is known and fixed, is a nonconvex optimization problem whose dimensionality increases with the number of UAVs. 

There have been several studies on the placement/trajectory optimization of UAVs for different objectives. For example, static placement of UAVs as mobile base stations to maximize network coverage \cite{mozaffari1, mozaffari2, shak}, and energy efficient placement of UAVs subject to ground terminal (GT) coverage or rate constraints are studied \cite{zeng2017energy, xuhind}. The problem of interference-aware joint trajectory and power control of multiple UAVs have been explored \cite{anotherzhangpap}. In \cite{caisecure}, the authors consider the UAV trajectory planning problem under communication security constraints.  In \cite{zeng2016throughput}, the authors propose a mobile UAV relaying method and jointly optimize throughput, trajectory and temporal power allocation. Power efficient deployment of UAVs as relays is studied in \cite{koyuncuspawc} for centralized and distributed UAV selection scenarios. Trajectory control of UAVs for maximizing the worst terrestrial user's spectral efficiency is studied in \cite{aydin}.  UAV placement problem have also been studied in several other different contexts including coordinated multi-point transmissions \cite{8675440}, public safety communications \cite{merwaday2015uav}, user-in-the-loop scenarios \cite{sncdrones}, adhoc network connectivity \cite{han2009optimization}; for numerous other applications and formulations, we refer the reader to the tutorial paper \cite{shakhatreh2019unmanned}.


Most of the above works consider optimizing the power or rate efficiency of the system or aim at maximizing the coverage. On the other hand, especially if the UAVs are tasked to collect data from low-power GTs, outage probability becomes an important performance measure. There are some studies that provide an outage probability analysis and the corresponding optimal deployments of UAV-based wireless networks. In particular, an outage probability analysis of a cooperative UAV network over Nakagami-$m$ fading channel has been presented in \cite{4109829}. Joint trajectory and power control problem for a single-UAV network with the outage probability performance measure has been considered in \cite{8068199}. In \cite{7412759}, the authors provide the outage analysis of a single-UAV network with underlaid device-to-device links. Optimum density of drone small cells subject to an outage constraint has been analyzed in \cite{zhang2019spectrum}. 

The aforementioned existing studies on outage-optimized UAV placement either consider the case of a single UAV or mainly consider a numerical approach to UAV location optimization. Formal analytical results on the optimal placement of a general number of UAVs and the resulting outage performance is not available in general. This motivates our work. Our main contributions in this context is summarized in the following.

\subsection{Summary of Main Contributions}
 We consider the optimal placement of UAVs serving as base stations to a density of GTs with the goal of minimizing the outage probability of GT transmissions. We follow an analytical approach to find optimal placement of UAVs, and derive upper bound and lower bounds on the optimal outage probability. Using the outage probability as a performance metric results in a fundamentally different cost function as compared to earlier literature on UAVs. Correspondingly, we also obtain fundamentally different results. In particular, for a wide class of GT densities, we show that the achievable outage probability decays exponentially with the number of available UAVs. Moreover, if the UAV minimum altitude constraint is sufficiently large, we show that all UAVs should collapse to a single location in an optimal deployment. As a result, a single multi-antenna UAV can be used in place of multiple single-antenna UAVs without loss of optimality. We present our results for Rayleigh fading as well as Rician fading with $\theta$-dependent $K$-factor and path-loss exponent \cite{Azari2018}, where $\theta$ is the elevation angle. 

For numerical optimization of UAV locations, we consider centralized as well as practical distributed algorithms. In particular, we use the particle swarm optimization (PSO) algorithm \cite{psomethod} to find UAV deployments with the best possible performance. In the context of UAV-based communications, applications of PSO to objective functions that are different than ours can be found in \cite{shakhatreh2017efficient, kalantari2016number}. In our case, PSO requires a decision center that has access to the entire GT density function and thus may not be practical. We thus also explore practical gradient descent based solutions where each UAV communicates with only its neighboring UAVs to find its optimal location in an autonomous manner. Moreover, for our distributed algorithm to work correctly, each UAV only has to know the user density in its immediate vicinity. Numerical results demonstrate that even with small sensing and communication range, the final deployments found through our distributed UAV deployment algorithm has little to no loss in performance relative to the deployments found using the baseline PSO algorithm.

One of the earliest works to apply gradient descent ideas to distributed placement problems is \cite{cortes1}. Since then, several variations including limited-range node interactions have been studied \cite{cortes2005spatially, gusrialdi2008coverage, guojafsensorlim, 8302930, 7951029, cass1, cass2}. In particular, a problem formulation that is similar to ours appears in the context of optimal control of mobile sensors \cite{cass1, cass2}. The objective of these two studies is to find control laws that maximize the cumulative event detection probabilities of the sensors. However, \cite{cass1, cass2} do not present analytical results on the structure and performance of optimal sensor deployments, and consider a different objective function.



Part of this work \cite{koyuncupimrc19} was presented at IEEE PIMRC 2019. Compared to \cite{koyuncupimrc19}, we provide here the proofs of the two main theorems of our study. Section \ref{secuavplacementalgos} of this paper, which introduces the centralized and the distributed UAV placement algorithms, and Section \ref{secextrician}, which provides extensions to a variant of Rician fading  for UAV systems, are both entirely new. The numerical comparison of the centralized and the distributed algorithms in Section \ref{secdistuavplacementalgoperf} is also new.

\subsection{Organization of the Paper}

The rest of this paper is organized as follows: In Section \ref{secsystemmodel}, we introduce the system model. In Section \ref{secoutfordifferentn}, we study the behavior of the minimum-possible outage probability with respect to the number of UAVs at a fixed altitude. In Section \ref{secoutfordifferenth}, we study optimal deployments of UAVs for different altitude constraints. In Section \ref{secuavplacementalgos}, we present our centralized and distributed UAV placement algorithms. In Section \ref{secextrician}, we extend our results to  Rician fading. In Section \ref{secnumerical}, we present numerical simulation results. Finally, in Section \ref{secconclusions}, we draw our main conclusions and discuss future work. Some of the technical proofs are provided in the appendices.
\section{System Model}
\label{secsystemmodel}
We consider multiple GTs at zero altitude and $n \geq 1$ UAVs. Mathematically, we assume that all the GTs are located at $\mathbb{R}^d$, where $d\in\{1,2\}$. The case $d=1$ is relevant when the GTs are located on a straight line on the ground, e.g. on a highway. We also assume that the GTs are distributed on $\mathbb{R}^d$ according to some probability density function $f:\mathbb{R}^d\rightarrow\mathbb{R}$.

In this paper, we consider a model where the UAVs are subject to a common minimum altitude constraint $h \geq 0$ due to governmental regulations or environmental constraints. In such a scenario, given that all the GTs are located at zero altitude, decreasing the altitude of any UAV decreases the GTs' access distance to the UAV, resulting in a better overall performance. We thus assume that all UAVs are located at the fixed altitude $h$. 

\begin{figure}
  \begin{center}
  \vspace{-20pt}
    \includegraphics[width=0.5\textwidth]{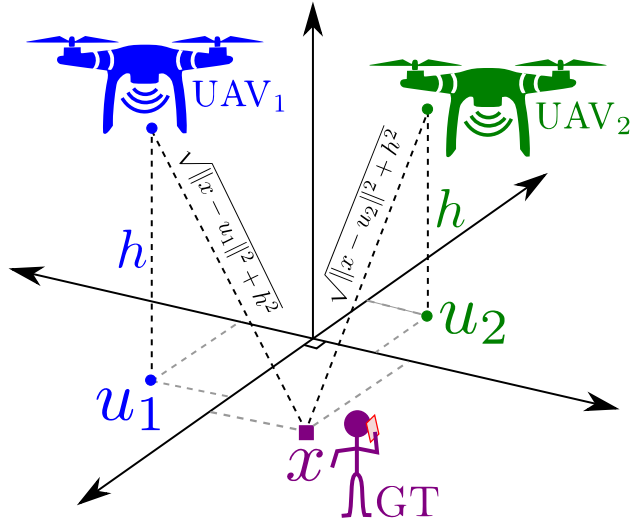}
  \end{center}
  \caption{ A network of two UAVs serving a GT.}
  \label{systemmodel}
\end{figure}

Let us now consider a GT at location $x\in\mathbb{R}^d$. Also, let $u_i\in\mathbb{R}^d,\,i=1,\ldots,n$ denote the locations of the UAVs projected to the ground (Hence, the actual location of UAV $i$ is given by $[u_i\,\,h]\in\mathbb{R}^{d+1}$). The GT can communicate with any one of the UAVs, all of which act as service providers. Fig. \ref{systemmodel} illustrates the system model for the special case of $n=d=2$.


Suppose that the GT wishes to communicate with rate $\rho$ bits/sec/Hz and transmits with fixed power $P$. Note that the distance between the GT at $x$ and UAV $i$ is given by $(\|x-u_i\|^2+h^2)^{\frac{1}{2}}$. Let $N_0$ be the noise power, the capacity of the channel between the GT and UAV $i$ is thus 
\begin{align}
C(x,u_i) \triangleq \log_2\left(1+\frac{A P |\beta_i|^2 }{N_0 (\|x-u_i\|^2 + h^2)^{\frac{r}{2}}}\right) \mbox{bits/sec/Hz}, 
\end{align}
where $r$ is the path loss exponent, $\beta_i$ is the fading coefficient between the GT and UAV $i$, and $A$ is a constant that depends on parameters such as operation frequency and antenna gain \cite{Azari2018}.

We assume $\beta_1,\ldots,\beta_n$ are independent and identically distributed as circularly-symmetric complex Gaussian random variable with unit variance. This results in the Rayleigh fading model, which is valid in a rich scattering environment, and is widely utilized in existing work on UAV communications. Extensions to other channel models that incorporate different practicalities of UAV communication systems will be provided later on. The GT transmission will be successful if the capacity of the channel carrying the GT's message is at least $\rho$. The probability of a failed transmission from a GT at $x$ to UAV $i$ is then given by the outage probability
\begin{align}
 P_{O_i}(x,u_i) & \triangleq \mathrm{P}\left(C(x,u_i) \leq \rho\right) \\
\label{pqwepoqiwepoqiweq} & =  \mathrm{P}\left( |\beta_i|^2 \leq \lambda (\|x-u_i\|^2 + h^2)^{\frac{r}{2}} \right).
\end{align}
Where $\lambda = \frac{N_0(2^{\rho}-1)}{A P}$. In order to evaluate the outage probability, we note the well-known fact that for a Rayleigh fading channel, $|\beta_i|^2$ is an exponential random variable with cumulative distribution function $F_{|\beta_i|^2}(x) = 1- e^{-x}$. Therefore, we obtain
\begin{align}
P_{O_i}(x,u_i)  =  1 - g(x,u_i),
\end{align}
where
\begin{align}
g(x,u_i) \triangleq \exp\left( -\lambda (\|x-u_i\|^2 + h^2)^{\frac{r}{2}} \right).
\end{align}
The transmitted GT data cannot be received by any one of the UAVs if an outage event occurs at all the UAVs. The overall outage probability of the entire system given the deployment $U = [u_1 \cdots u_n]$ of UAVs is thus
\begin{align}
P_{O}(x,U)  =  \prod_{i=1}^n \left[1 - g(x,u_i) \right].
\end{align}
Averaging out the GT density function $f$, we arrive at our objective function
\begin{align}
\label{objectivefunc}
P_{O}(U)  =  \int \prod_{i=1}^n \left[1 - g(x,u_i) \right] f(x)\mathrm{d}x.
\end{align}
Throughout the paper, all unspecified integration domains are $\mathbb{R}^d$. The problem is then to find optimal UAV locations $U^{\star} \triangleq \arg\min_U P_O(U)$ that minimize the outage probability, and the corresponding minimum-possible outage probability $P_O(U^{\star})$. In the following, we first consider the behavior of $P_O(U^{\star})$ for different values of the number of UAVs $n$ for a fixed UAV altitude $h$. We then consider the behavior of $P_O(U^{\star})$ as a function of the constraint $h$ on the UAV altitudes. 

\section{Optimal Outage Probability for Different Number of UAVs}
\label{secoutfordifferentn}
In this section, we study how the optimal outage probability varies with respect to the number of UAVs for a fixed $h$, and the corresponding optimal deployments.  In this context, finding the exact minimizers of (\ref{objectivefunc}) appears to be a hopeless task for a general $n$. Even the case of a single UAV $n=1$ results in a non-trivial optimization problem, for which there appears to be a no closed-form solution in general. Still, one can find the optimal positioning of a single UAV for unimodal densities that are defined in the following. 
\begin{definition}
\label{unimoddef}
We call a univariate density function $f:\mathbb{R}\rightarrow\mathbb{R}$ unimodal with center $\mu\in\mathbb{R}$ if $f$ is non-decreasing on $(-\infty,\mu]$, $f$ is non-increasing on $[\mu,\infty)$, and $f(\mu+c) = f(\mu-c),\,\forall c\in\mathbb{R}$. We call a bivariate density function $f:\mathbb{R}^2\rightarrow\mathbb{R}$ unimodal with center $\mu = [\begin{smallmatrix} \mu_1 \\ \mu_2 \end{smallmatrix}]\in\mathbb{R}^2$ if for all $y$, the function $f_1(x) = f(x,y)$ is unimodal with center $\mu_1$, and for all $x$, the function $f_2(y) = f(x,y)$ is unimodal with center $\mu_1$.
\end{definition}
A circularly-symmetric Gaussian random vector or uniform random vectors are both examples of unimodal densities. We have the following theorem.
\begin{theorem}
\label{oneuavtheorem}
Let $n=1$, and suppose that $f$ is unimodal with center $\mu$. Then, for any $h \geq 0$, an optimal deployment of the single UAV is given by $U^{\star} = \mu$.
\end{theorem}
\begin{proof}
See Appendix \ref{oneuavtheoremproof}.
\end{proof}
This verifies the intuition that if the density is ``symmetric'' around $\mu$, then the optimal deployment of the single UAV should be $\mu$.

For the case of a general $n \geq 1$, we provide general upper and lower bounds on the outage probability that hold for any number of UAVs and different GT densities. The bounds allow us to obtain estimates on the achievable outage probability without going through the time-consuming numerical optimization methods. They also allow us to obtain the asymptotic behavior of the outage probability as $n\rightarrow\infty$.

We first present a general upper bound via the following theorem. As defined in Section \ref{secsystemmodel}, let $P_O(U^{\star})$ denote the minimum possible outage probability provided by an optimal deployment. Let $X$ represent a random variable whose probability density equals the GT density $f$.
\begin{theorem}
\label{randdeploy}
The following upper bound holds for arbitrary random variables $U_1,\ldots,U_n$:
\begin{align}
\label{qoiweuoqiweuqw}
P_O(U^{\star}) \leq \mathrm{E}_{X,U_1,\ldots,U_n}\left[\prod_{i=1}^n (1 -g(X,U_i)) \right].
\end{align}
In particular, suppose $U_1,\ldots,U_n$ have the same density as some random variable $U$, are mutually independent, and also independent of $X$. Then, we have
\begin{align}
\label{qoiweuoqiweuqw2}
P_O(U^{\star}) \leq \int \left(\mathrm{E}_{U}\left[ 1- g(x,U) \right] \right)^n f(x)\mathrm{d}x.
\end{align}
\end{theorem}
\begin{proof}
The proof (\ref{qoiweuoqiweuqw}) relies on the following random deployment argument: We assume that the location of UAV $i$ is the random variable $U_i$. The expected outage probability with this random deployment scenario is the right side of (\ref{qoiweuoqiweuqw}). It is guaranteed that there exists at least one deployment that achieves the expected performance, and this proves (\ref{qoiweuoqiweuqw}). Result (\ref{qoiweuoqiweuqw2})  follows immediately from (\ref{qoiweuoqiweuqw}) by independence of random variables. 
\end{proof}
In particular, the theorem shows via (\ref{qoiweuoqiweuqw2}) that in general, the outage probability decays at least exponentially with the number of available UAVs provided that the support of $f$ is finite. In fact, setting $U = 0$ to be deterministic, we have $P_O(U^{\star}) \leq \mathrm{E}[(1- g(X,0))^n] \leq \mathrm{E}[(1- e^{-\lambda\|X\|^r})^n]$. Given that the support of $f$ is contained on a ball with radius $R$, we then have $P_O(U^{\star}) \leq (1- e^{-\lambda R^r})^n$, proving the exponential rate of decay as desired. On the other hand, the upper bound in (\ref{qoiweuoqiweuqw2}) is not tight in general for a given finite $n$. An interesting problem in this context that we shall leave as future work is the optimization of the upper bound (\ref{qoiweuoqiweuqw}) with respect to the densities $U_1,\ldots,U_n,X$ for tighter upper bounds.

We now provide lower bounds. A very simple lower bound that holds for arbitrary densities is the following:
\begin{theorem}
\label{lowerboundone}We have
\begin{align}
\label{triviallb}
 P_{O}(U^{\star}) \geq (1-e^{-\lambda h^r})^n.
 \end{align}
\end{theorem}
\begin{proof}
The bound follows once we use the estimate $g(x,u_i) = \exp( -\lambda (\|x-u_i\|^2 + h^2)^{\frac{r}{2}} ) \leq \exp( -\lambda  h^r )$ 
 in  (\ref{objectivefunc}).
\end{proof}
The bound (\ref{triviallb}) is unfortunately not useful for the case of unmanned ground vehicles (UGVs) $h=0$. For this scenario, using somewhat more sophisticated methods, one can still obtain a lower bound with an exponential rate of decay.
\begin{theorem}
\label{triviallbh0}
Let $h=0$. There is a constant $c \geq 0$ such that for any $n$, we have $P_{O}(U^{\star}) \geq c^n$. 
\end{theorem}
\begin{proof}
See Appendix \ref{proofoftriviallbh0}.
\end{proof}

Theorems \ref{randdeploy}, \ref{lowerboundone}, and \ref{triviallbh0} together show that for any fixed $h$, the best-possible decay of the outage probability is exponential with respect to the number of UAVs.

\section{Optimal Deployments for Different UAV Altitude Constraints}
\label{secoutfordifferenth}
We now consider the behavior of the optimal UAV deployments and the corresponding optimal outage probabilities for different UAV altitude constraints $h$. By Theorem \ref{oneuavtheorem}, we have shown that, for one UAV and a unimodal GT density, the optimal UAV location is at the center of the GT density. For example, for a uniform distribution on $[0,1]$, the optimal UAV location is at $0.5$ for every UAV altitude. With more UAVs, one would expect the UAVs to be evenly distributed to $[0,1]$ to provide an even coverage of the GT domain $[0,1]$. This is the case, indeed, if one wishes to place the UAVs so as to minimize the average GT power consumption to guarantee zero-outage transmission \cite{ouruav2}, or in many existing studies on UAV location optimization \cite{zeng2016wireless}. 

A counterintuitive phenomenon occurs in the case of our objective function of minimizing the outage probability. In fact, during our numerical experiments, we have observed that after a certain finite altitude, all UAVs collapse to a single location in an optimal deployment. Here, we prove a slightly weaker claim that, in an optimal deployment, all UAVs converge to a common location as $h\rightarrow\infty$. We write $f(x) \sim g(x)$ as a shorthand notation for the asymptotic equality $\lim_{x\rightarrow\infty} f(x)/g(x) = 1$. 

\begin{theorem}
\label{asymptottheorem}
Let
\begin{align}
u^{\star}(h) \triangleq \arg\max_u \int g(x,u)f(x)\mathrm{d}x.
\end{align}
As $h\rightarrow\infty$, we have the asymptotic equality
\begin{align}
\label{qiwepoqwueqw}
\left(1 - P_O(U^{\star})\right) \sim n  \int g(x,u^{\star}(h))f(x)\mathrm{d}x.
\end{align}
An optimal deployment $U^{\star}$ that achieves (\ref{qiwepoqwueqw}) is where all the $n$ UAVs are located at $u^{\star}(h)$ for a given $h$.
\end{theorem}
\begin{proof}
See Appendix \ref{proofofasymptottheorem}.
\end{proof}

The theorem provides a precise asymptotic formula for the optimal outage probability as $h\rightarrow\infty$ and the optimal UAV locations that achieve this outage probability. In particular, placing all UAVs to the point $\lim_{h\rightarrow\infty} u^{\star}(h)$ is asymptotically optimal.

Intuition suggests that when one has a uniform distribution of users over an area, and multiple base stations to serve these users, the optimal deployment of base stations should be more or less uniform over the area. From this perspective, the conclusion of Theorem \ref{asymptottheorem} is very counter-intuitive: It suggest that in an optimal deployment, all base stations (UAVs) should collapse to a single location. Note that in an ordinary network consisting of base stations and users, by design, every user is assigned to one and only one base station. One should then distribute the base stations uniformly over the area so that the access distance of each user is minimized, i.e. every user on the area can find one base station nearby. On the other hand, in our scenario, every base station serves every user, as the transmitted signal of a user can be decoded by any UAV that is not in outage. Moreover, when the altitude is high, the probability of outage becomes much higher. In such a scenario, it makes more sense for a user to have \textit{all UAVs within reasonable distance} (and thus benefit from the spatial diversity provided by the UAVs) rather than to have \textit{one UAV within close distance}. The requirement of having all UAVs within reasonable distance for all users means that all UAVs should be located more or less in the area center.


An important design implication of Theorem \ref{asymptottheorem} is that, for a large UAV altitude constraint, one can just use a single UAV with multiple antennas (and use selection diversity), as compared to multiple UAVs with a single antenna. In fact, since it is physically impossible to place multiple single-antenna UAVs to a single location, one \textit{must} use a multi-antenna UAV to minimize the outage probability for a given maximum number of selectable UAV antennas.

A special case is when the density $f$ is unimodal with center $\mu$. In such a scenario, we have $u^{\star}(h) = \mu,\,\forall h$, and the UAV locations converge to $\mu$ as $h\rightarrow\infty$ by Theorem \ref{asymptottheorem}. For example, if $f$ is one-dimensional Gaussian with mean $\mu$ and variance $\sigma^2$, $\lambda = 1$ the outage probability when all UAVs are located at $\mu$ is given by
\begin{align}
 P_O([\mu\cdots\mu])  & = \int (1-g(x,\mu))^n \frac{1}{\sqrt{2\pi\sigma^2}}e^{-\frac{(x-\mu)^2}{2\sigma^2}}\mathrm{d}x \\
\label{gaussasymptote} & = \sum_{k=0}^n \binom{n}{k} \frac{(-1)^k e^{-kh^2}}{\sqrt{1+2k\sigma^2}}\mathrm{erf}\left(\sqrt{\frac{1+2k\sigma^2}{2\sigma^2}} \right).
\end{align}
For large $h$, we expect (\ref{gaussasymptote}) to be a good approximation on the optimal outage probability. As a two-dimensional example, for a uniform distribution on $[0,1]^2$, we have
\begin{align}
\label{pqowepoqwepoqwie}
P_O([\begin{smallmatrix} 0.5 & \cdots & 0.5 \\ 0.5 & \cdots &  0.5 \end{smallmatrix} ] )  =  1 + \frac{\pi}{4} \sum_{k=1}^n \binom{n}{k} \frac{(-1)^k}{k e^{kh^2}}\mathrm{erf}^2(\sqrt{k})
\end{align}
to approximate the optimal outage probability at large $h$.


\section{Algorithms for UAV Placement}
\label{secuavplacementalgos}
In the two previous sections, we have presented several theorems that provide the optimal UAV placement under different scenarios and the corresponding outage probabilities. In this section, we present several numerical algorithms that provide the optimal UAV placement in general. We explore centralized solutions that can provide globally-optimal solutions as well as distributed solutions that are amenable to real-time implementation.

\subsection{Centralized Solution: Particle Swarm Optimization}
\label{psoapproach}
First, we seek an algorithm that can provide the optimal UAV deployment without any limitation on system complexity. According to the formulation of the objective function in (\ref{objectivefunc}), the optimal deployment depends on the user density function. Hence, we assume the existence of a decision center that perfectly knows the user density. The decision center will then solve the problem in (\ref{objectivefunc}) via a powerful global optimization algorithm and tell each UAV its optimal location. The UAVs may then move to their respective locations. 

We have used the particle swarm optimization (PSO) method at the decision center. PSO is a population-based stochastic iterative algorithm for solving non-linear optimization problems \cite{psomethod}. In general, population-based optimization algorithms such as PSO are known to outperform the simpler gradient descent like approaches. In fact, the existence of multiple candidate solutions (population members) help to avoid locally optimal solutions. We will later compare the outage probabilities obtained using the PSO method with those provided by gradient-descent based distributed solution. In general, our PSO-based centralized solution will outperform the distributed gradient descent-based solution, albeit at the expense of increased system complexity. 

For completeness, we provide the details of the PSO algorithm in the following: In general, PSO initializes multiple agents/particles at arbitrary positions in the optimization space and searches for the global optima in an iterative manner. Let the subscript $i$ denote the particle index, and the superscript $t$ denote the iteration index. Each particle evaluates the cost function at its own location at each iteration and updates its location ($U_i^{t+1}$) based on its current direction/velocity ($V_i^{t}$), its own best record ($\mathrm{PB}_i^{t}$) and the global best record ($\mathrm{GB}^{t}$) via the following equations: 
\begin{align}
\label{veloupdate}
V_i^{t+1} & = \omega V_i^{t} + C_1 r_1 (\mathrm{PB}_i^{t} - U_i^{t}) + C_2 r_2 (\mathrm{GB}^{t} - U_i^{t}), \\
\label{locupdate}
U_i^{t+1} & = U_i^{t} + V_i^{t+1}.
\end{align}
Here, $V_i^{t}$ is the velocity of Particle $i$ at Iteration $t$ and $U_i^{t}$ is the position of Particle $i$ at Iteration $t$. Also, $\omega$ is the inertia coefficient that determines the speed of exploration. In particular, a higher $\omega$ results in a wider search area and hence more exploration and a lower $\omega$ results in less exploration. In our work, we decrease $\omega$ as the number of iterations increases so as to guarantee the convergence of the algorithm. $\mathrm{PB}_i^{t}$ is the best personal record (the location with lowest cost) for Particle $i$ at Iteration $t$ and $\mathrm{GB}_{t}$ is the global best record at iteration $t$. 

Scalars $r_1$ and $r_2$ are chosen independently and uniformly at random for each particle at each iteration. Deterministic parameters $C_1$ and $C_2$ then control the so-called ``exploration'' and ``exploitation'' behavior of the algorithm. For example, with a larger $C_1$, particles tend to converge to their own personal best solutions, resulting in exploitation of local optima. On the other hand, with a larger $C_2$, particles tend to converge to the global best. For such convergence to occur, the particles tend to travel larger distances, which translates to an increased exploration of search space. In update step, first the velocity is updated based on (\ref{veloupdate}), and then the location of particle is updated based on (\ref{locupdate}). Note that for the application to our UAV placement problem, the location of one PSO particle corresponds to the entire UAV deployment.



The pseudo code of the entire PSO procedure is then summarized in Algorithm \ref{psoalgo}.

\begin{algorithm}
\caption{Particle Swarm Optimization for UAV Placement}
\begin{algorithmic}[1]
\State Parameters: Number of particles $N$, Number of iterations $T$.
\State Initialize particle positions $U_i^1,\,i=1,\ldots,N$.
\State Set $\mathrm{PB}_i^{1} = U_i^1$, and $\mathrm{GB}^1 = \arg\min_{\mathrm{PB}_i^1} P_O(\mathrm{PB}_i^{1})$, where $P_O$ is defined in (\ref{objectivefunc}).
\For{$t = 1$ to $T$}
\For{$i = 1$ to $N$}
\State  Calculate particle velocity $V_i^{t+1}$ according to (\ref{veloupdate}).
\State  Update particle position $U_i^{t+1}$ according to (\ref{locupdate}).
\State {\bf if $P_O(U_i^{t+1}) < P_O(\mathrm{PB}_i^t)$ then $\mathrm{PB}_i^{t+1} = U_i^{t+1}$ else $\mathrm{PB}_i^{t+1} = \mathrm{PB}_i^{t}$}.
\EndFor
\State Set $\mathrm{GB}^{t+1} = \arg\min_{\mathrm{PB}_i^{t+1}} P_O(\mathrm{PB}_i^{t+1})$.
\EndFor
\end{algorithmic}
\label{psoalgo}
\end{algorithm}

\subsection{Distributed Solution: Gradient Descent Optimization} 
Due to its population-based nature, we expect the PSO based approach of Section \ref{psoapproach} to provide the optimal or close-to-optimal performance given sufficiently many particles.
We will also verify this intuition through our numerical simulation results in Section \ref{secnumerical}. On the other hand, as discussed in Section \ref{psoapproach}, PSO requires centralized processing and perfect knowledge of the user density over the entire area of interest, which may not be practical. We thus consider the design of  distributed UAV location optimization algorithms that can work with limited user density information. We will utilize the  performance achieved by the PSO algorithm as a baseline for such practical distributed algorithms.

We follow the standard gradient descent idea \cite{cortes1} to come up with a distributed UAV location optimization scheme. The gradients of the cost function in (\ref{objectivefunc}) can be calculated to be:
\begin{align}
\label{uavgrad}
    G_i \triangleq \frac{\partial P_O(U)}{\partial u_i} = \int  h(x,u_i) \textstyle\prod_{\substack{j=1 \\ j\neq i}}^n (1-g(x,u_j)) f(x)\mathrm{d}x,i=1,\ldots,n,
\end{align}
where
\begin{align}
h(x,u_i) \triangleq -\frac{\partial g(x,u_i) }{\partial u_i} =r \lambda (x-u_i) (\|x-u_i\|^2 + h^2)^{\frac{r}{2}-1} \exp\left( -\lambda (\|x-u_i\|^2 + h^2)^{\frac{r}{2}} \right).
\end{align}
UAV $i$ may then perform the gradient descent update $u_i \leftarrow u_i - \eta \frac{\partial P_O(U)}{\partial u_i}$ for some learning parameter $\eta > 0$. Unfortunately, to calculate the gradient (\ref{uavgrad}), each UAV needs to know the location of all other UAVs, which is not practical. Therefore, we consider a system where each UAV can communicate only with the UAVs within a certain given distance $D_c > 0$. This results in the approximate gradient
\begin{align}
   G_i' = \int h(x,u_i) \textstyle\prod_{\substack{j=1,\, j\neq i \\ \|u_j - u_i\| \leq D_c}}^n (1-g(x,u_j)) f(x)\mathrm{d}x.
\end{align}
In addition, sensing the user density of the whole area of interest may not be feasible in terms of energy or computational complexity. As a result, we assume that each UAV can only know the user density function within a certain distance $D_s > 0$. In other words, we assume that UAV $i$ only knows the user density function within the disk $\mathcal{B}(u_i,D_s)$ of radius $D_s > 0$ centered at $u_i$. The approximate gradient at UAV $i$ can then be expressed as
\begin{align}
\label{finalgradientform}
   G_i'' = \int_{\mathcal{B}(u_i,D_s)} h(x,u_i) \textstyle\prod_{\substack{j=1,\, j\neq i \\ \|u_j - u_i\| \leq D_c}}^n (1-g(x,u_j)) f(x)\mathrm{d}x.
\end{align}

Our entire distributed gradient-descent based UAV location optimization solution is summarized in Algorithm \ref{graddescentalgo}. We note that, in Algorithm \ref{graddescentalgo}, the learning parameter $\eta$ may be modified to increase the speed of convergence or to prevent oscillations/divergence.
In practice, all UAVs may agree on an initial $\eta$ and modify it at each iteration according to a predetermined policy. Overall, the algorithm can be implemented in a distributed manner with each UAV communicating with only the UAVs within distance $D_c$. Moreover, each UAV only needs to know the user density function within the distance $D_s$. We will observe in the next section that even with small values of $D_c$ and $D_s$, the performance with Algorithm \ref{graddescentalgo} is very close to the performance with the PSO-based Algorithm \ref{psoalgo}. 

\begin{algorithm}
\caption{Gradient Descent Optimization for Distributed UAV Placement}
\begin{algorithmic}[l]
\State Parameters: Number of iterations $T$.
\State Initialize UAV locations $u_1,\ldots,u_n$.
\For{$t = 1$ to $T$}
\For{$i = 1$ to $n$}
\State UAV $i$ finds neighbors within distance $D_c$.
\State UAV $i$ calculates its local gradient according to (\ref{finalgradientform}).
\State UAV $i$ updates its location according to $u_i \leftarrow u_i - \eta G_i''$
\EndFor
\EndFor
\end{algorithmic}
\label{graddescentalgo}
\end{algorithm}

\section{Extensions to a Practical Rician Fading Model}
\label{secextrician}

In this section we extend our results to a modified Rician fading model that is specifically tailored for UAV communication systems \cite{Azari2018}. To introduce the model, let $\theta_i = \tan^{-1}(\frac{h}{\|x-u_i\|})$ denote the elevation angle between a UAV at $u_i$ and a user at $x$. The model relies on path loss exponents that depend on the elevation angles through $r(\theta) = a_1 P_{LOS}(\theta)+b_1$, where  $P_{LOS}(\theta) = 1/(1+a_2 e^{-b_2(\theta-a_2)})$ is the line of sight probability, and $a_1,\,b_1,\,a_2$ and $b_2$ are determined by the environment
characteristics and the transmission frequency. The model also relies on an angle-dependent $K$-factor $K(\theta) = a_3 e^{b_3 \theta}$, where $a_3$ and $b_3$ are constants. The channel gain is then expressed via the Rician PDF
\begin{equation}
\label{ricianpdfexpre}
f_{|\beta_i|^2}(\omega) = (K(\theta)+1)\exp{\left(-K(\theta)-(K(\theta)+1)\omega\right)}I_{0}\left(2\sqrt{K(\theta)(K(\theta)+1)\omega}\right).
\end{equation}

Using (\ref{pqwepoqiwepoqiweq}) and (\ref{ricianpdfexpre}), the outage probability is:
\begin{align}
 P_{O_i}(x,u_i) & = \mathrm{P}\left( |\beta_i|^2 \leq \lambda (\|x-u_i\|^2 + h^2)^{r(\theta_i)/2} \right)  \\
\label{qfunctionhere} & =  1 - \mathrm{Q}\left(\sqrt{2K(\theta_i)},\sqrt{2\lambda (K(\theta_i)+1)(\|x-u_i\|^2 + h^2)^{r(\theta_i)/2}}\right).
\end{align}
Here, the first order Marcum $Q$ function is defined as
\begin{align}
\label{marcumqdef333}
Q(a,b) = \int_{b}^{\infty} x \exp \left(-\frac{x^2+a^2}{2}\right) I_0(ax) \mathrm{d}x,
\end{align}
 where $I_\nu(x)$ is the modified Bessel function of order $\nu$.
Now, let $w(x,u_i)$ denote the $Q$-function on the right hand side of (\ref{qfunctionhere}). 
Averaging out the user density function, the outage probability is
\begin{align}
\label{objectivefuncRician}
P_{O}(U)  =  \int \prod_{i=1}^n \left[1 - w(x,u_i) \right] f(x)\mathrm{d}x.
\end{align}

In the following, we provide a closed form solution for the optimal placement of UAVs for one UAV at an arbitrary altitude, and arbitrary number of UAVs at high altitudes. 

\begin{theorem}
\label{oneuavtheoremRician}
Let $n=1$, and suppose that $f$ is unimodal with center $\mu$. Then, for the Rician fading model, for any $h \geq 0$, an optimal deployment of the single UAV is given by $U^{\star} = \mu$.
\end{theorem}
\begin{proof}
See Appendix \ref{appendixRician}.
\end{proof}

 Similarly to Theorem \ref{oneuavtheorem}, if the density is symmetric around $\mu$, then the optimal deployment of the single UAV should be $\mu$ for the Rician fading model. We can also verify that the collapsing phenomenon in Theorem \ref{asymptottheorem} extends to Rician fading as well.

\begin{theorem}
\label{largeHtheoremRician}
Suppose that $f$ is unimodal with center $\mu$. For the Rician fading model, at high altitudes ($h\rightarrow\infty$), all UAVs should be located at $\arg\max_u \int w(x,u)f(x)\mathrm{d}x$.
\end{theorem}
\begin{proof}
See Appendix \ref{appendixRician2}.
\end{proof}

Collapsing appears to be a general phenomenon that may apply to a yet wider class of fading distributions. We leave a detailed further study as future work.


\section{Numerical Results}
\label{secnumerical}
In this section, we provide numerical simulation results that confirm our analytical findings.

\subsection{Verification of Analytical Results on Optimal Deployments}

In Fig. \ref{figuniform}, we compare the optimal outage probability provided by the PSO algorithm with the analytical formula (\ref{pqowepoqwepoqwie}) for different values of the UAV altitude $h$, and path loss exponent $r = 2$. The horizontal axis represents the number of UAVs, and the vertical axis represents the logarithm of the optimal outage probability. We can observe that the logarithm of the outage probability decays linearly with the number of UAVs. Hence, the outage probability decays exponentially with $n$, verifying Theorems \ref{randdeploy}, \ref{lowerboundone}, and \ref{triviallbh0}.

\begin{figure}
  \begin{center}
    \includegraphics[width=0.7\textwidth]{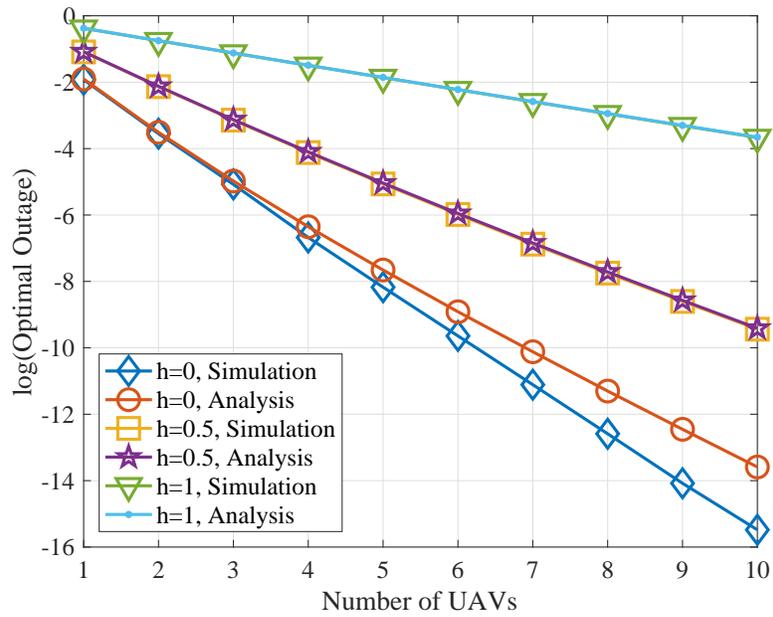}
  \end{center}
  \vspace{-5pt}
  \caption{ Outage probabilities for $f(x) = 1(x\in[0,1]^2)$.}
  \label{figuniform}
\end{figure} 

As $h$ increases, the analytical formula in (\ref{pqowepoqwepoqwie}) provides a very good approximation of the optimal outage probability derived with PSO algorithm (Indicated by 'Simulation' in the figure). In fact, for $h\in\{0.5,1\}$, the analysis is almost indistinguishable from the simulations.

\begin{figure}
  \begin{center}
  \vspace{-20pt}
    \includegraphics[width=0.7\textwidth]{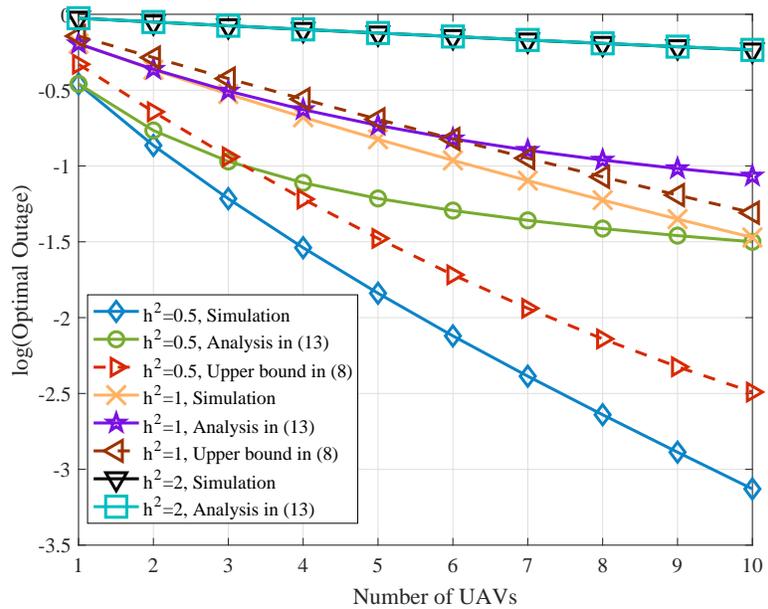}
  \end{center}
  \vspace{-5pt}
  \caption{ Outage probabilities for a standard normal GT density.}
  \label{fignormal}
\end{figure}

In Fig. \ref{fignormal}, we consider a standard normal distribution for the GT density, and compare the outage probabilities provided by the PSO algorithm with the analytical formula in (\ref{gaussasymptote}) for $r=3$. We can similarly observe the exponential decay of the outage probability with respect to the number  of UAVs. Also, for any number of UAVs, the analysis matches the simulations perfectly when $h^2 = 2$. On the other hand, the mismatch for lower altitudes is more pronounced compared to the case of the uniform distribution in Fig. \ref{figuniform}. As a result, we have also included the upper bounds derived using (\ref{qoiweuoqiweuqw2}).

We chose the random variable $U$ to be a standard normal distribution, which provided the best upper bound among all other choices for $U$ that we have considered for large $n$. We can observe that the upper bound in (\ref{qoiweuoqiweuqw2}) provides a better prediction of the optimal performance when $n$ is large. Yet, there is still a gap between the estimated and the actual performance. This shows the necessity of better bounds on the outage probability, a topic that will be studied as future work.

   \begin{figure}[!ht]
     \subfloat[Rayleigh fading.
     \label{figcollapse}]{%
       \hspace{-10pt}\includegraphics[width=0.55\textwidth]{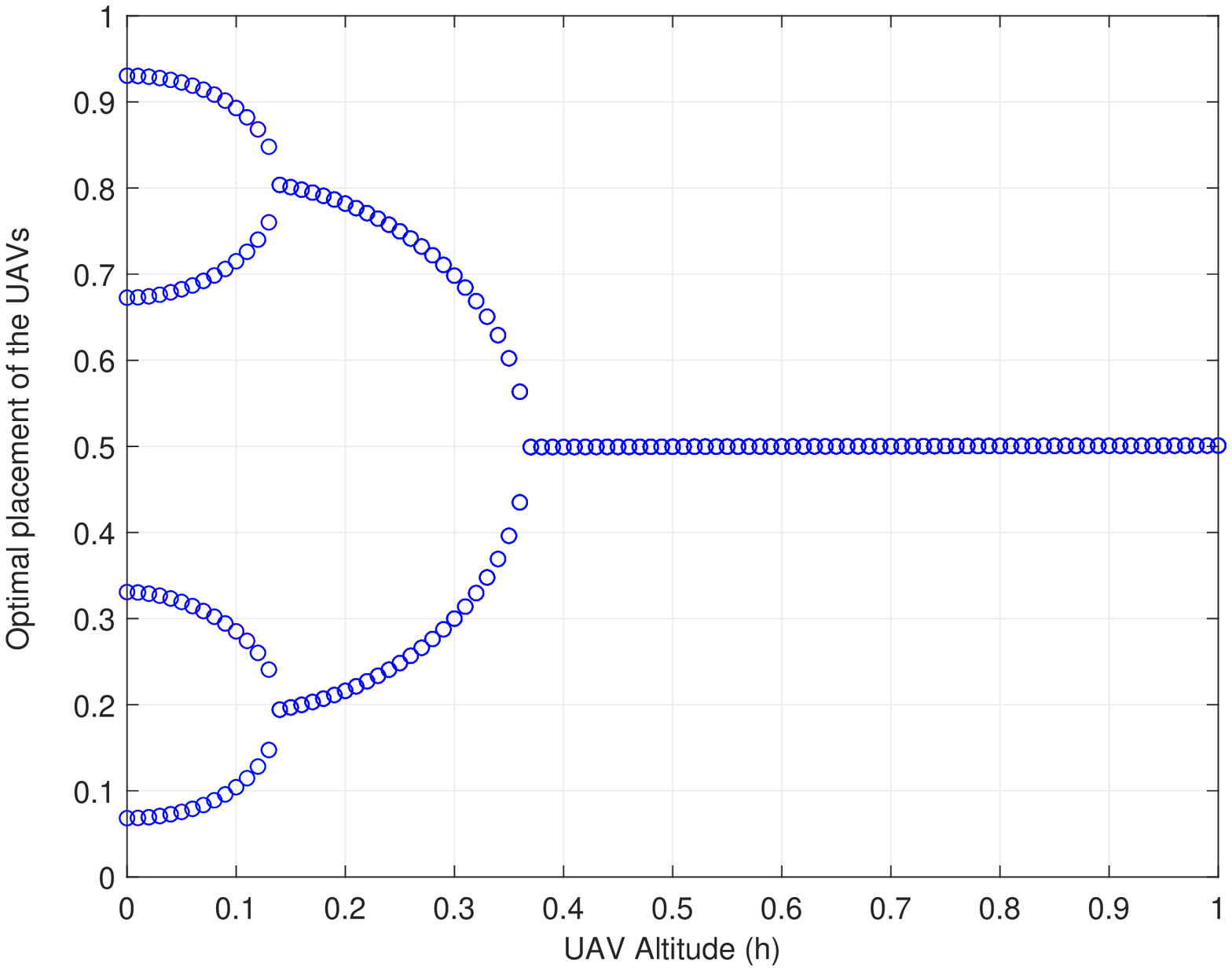}
       \hspace{-10pt}
     }
     \hfill
     \subfloat[Rician fading.
     \label{figcollapseRice}]{
       \hspace{-20pt}\includegraphics[width=0.55\textwidth]{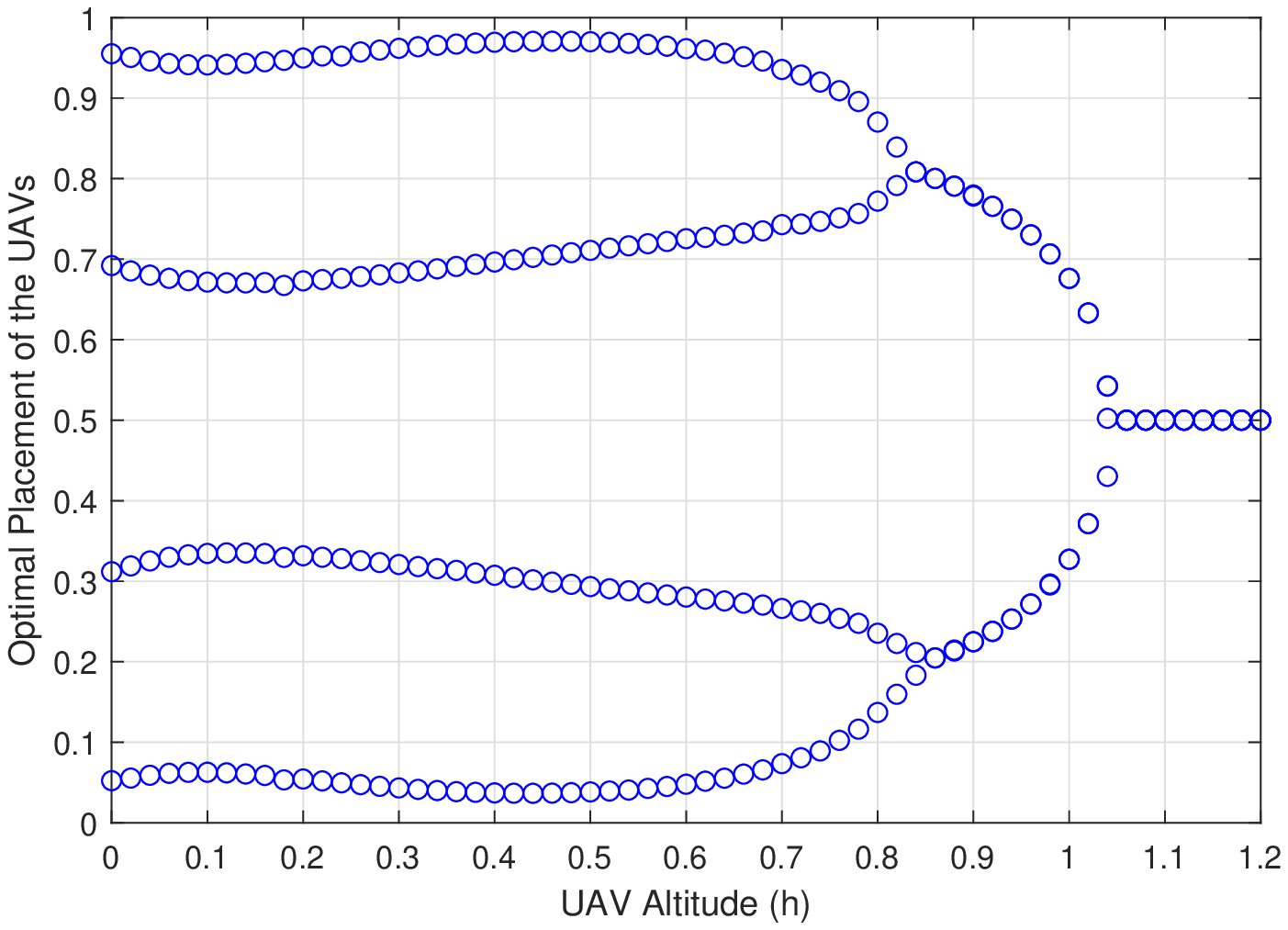}
       \vspace{-20pt}
     }
     \caption{ Optimal locations of $4$ UAVs for a uniform density on $[0,1]$ }
     \label{fig4sub}
   \end{figure}

In Fig. \ref{figcollapse} and \ref{figcollapseRice}, we illustrate the collapse of the optimal UAV locations to a single location as the UAV altitude constraint $h$ increases for both Rayleigh and Rician model. We have considered the choice $r=2$. The horizontal axis represents $h$, and the vertical axis represents the optimal UAV locations. For the Rayleigh model, we can observe that at $h=0$, the optimal deployment is roughly given by $U^{\star} = [0.08 \,\, 0.33 \,\,0.66 \,\,     0.92]$. As $h$ is increased to around $0.15$, the four UAVs first collapse to two distinct locations. After $h \geq 0.4$, the optimal locations for all four UAVs is equal to $0.5$, as Theorem \ref{asymptottheorem} suggests. Note that Theorem \ref{asymptottheorem} shows that the UAV locations will converge to $0.5$ as $h\rightarrow\infty$. Here, we observe the stronger phenomenon that the UAVs should be located at $0.5$ for all sufficiently large altitudes. A similar pattern can be observed for the UAV trajectories derived based on Rician fading model. 


\vspace{20pt}
\subsection{Performance of Distributed UAV Placement Algorithms}
\label{secdistuavplacementalgoperf}


We now compare the practical distributed UAV placement method in Algorithm \ref{graddescentalgo} with the centralized PSO-based method in Algorithm \ref{psoalgo}. The simulation parameters are as follows: A suburban area is considered with carrier frequency $f_c = 2.4$ GHz. For the $K$-factor and the path loss exponent, the parameters are set as  $a_1 = -1.5$, $b_1 = 3.5$, $a_2 = 4.88$, $b_2 = 0.43$,  $a_3=5$, $b_3 = \frac{2}{\pi}\log 3$, $\frac{A P}{N_0} = 75dB$. These are the same choice of parameters as in \cite{Azari2018}.

In Fig. \ref{distalgofig1}, we show the numerical simulation results for the one-dimensional uniform density function $f(x) = 10^{-3}(x\in[0,1000])$ at altitude $h = 500$. All distances are in meters. The UAV locations are initialized independently and uniformly at random on $[0,1000]$ in Algorithm \ref{graddescentalgo}. This is a valid assumption provided that we have prior information about the boundaries of the area of interest. The shown results are averages over many initialization of UAV locations.

\begin{figure}
  \begin{center}
  \vspace{-20pt}
    \includegraphics[width=0.7\textwidth]{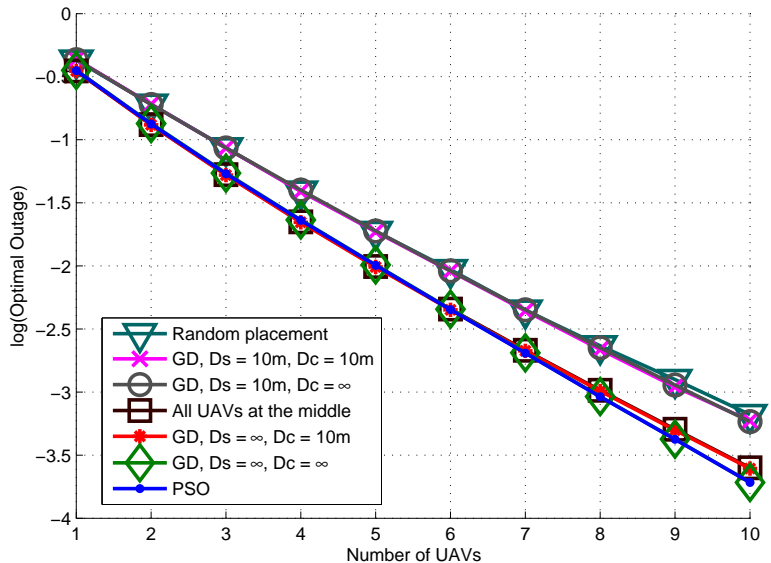}
  \end{center}\vspace{-10pt}
  \caption{ Comparison of UAV deployment algorithms for a one-dimensional uniform density at $h = 500m$.}\vspace{-10pt}
  \label{distalgofig1}
\end{figure}


We can observe that PSO provides the best-possible performance. Moreover, ``ordinary'' gradient descent ($D_s=D_c=\infty$) provides almost the same performance as PSO. As the sensing and/or the communication range is decreased, the performance typically deteriorates. Interestingly, we can also observe that decreased communication and sensing range becomes particularly harmful as one employs more UAVs in the system. This is expected, since with more UAVs, a larger number of factors are ignored while transitioning from the ``true'' gradient equation (\ref{uavgrad}) to its approximate form (\ref{finalgradientform}).  

Useful insight on the effect of communication and sensing ranges can be obtained by investigating the asymptotes $D_s \rightarrow \{0,\infty\}$ and/or $D_c \rightarrow \{0,\infty\}$. In particular, for $D_s = \infty$ and $D_c \rightarrow 0$, any given UAV will be unable to ``see'' the other UAVs, and thus collapse to the optimal location of a one-UAV system (which for this case is the mean of distribution, because the uniform density is a unimodal function based on Definition 1) after sufficiently many iterations. In fact, we can observe from Fig. \ref{distalgofig1} that the outage probability with gradient descent for $D_s = \infty$ and $D_c=10m$ is almost the same as the outage probability for the case where all UAVs are located at the middle.

Another extreme case is $D_s,  D_c \rightarrow 0$. In this case, as $D_c \rightarrow 0$, a given UAV will ``think'' it is the only UAV in the system due to lack of communication with the other UAVs. Moreover, as $D_s \rightarrow 0$, each UAV will estimate the user density to be uniform over a set of radius $D_s$ centered around itself. Since the optimal deployment of a single UAV for a uniform distribution is at the center of the area of interest, the modified gradient descent update equations for $D_s,  D_c \rightarrow 0$ imply that the UAVs will not move from their initial positions. A similar phenomenon occurs even if $D_c$ is non-zero provided that $D_s$ stays close to zero. In this case, although all UAVs can communicate to each other, each UAV thinks that it is at the center of a ``very narrow'' area of interest, and all the other UAVs are outside the area. The gradient descent updates imply little to no movement for each UAV. As a result, at convergence, each UAV ends up roughly in the same place that it has been initialized. In such a scenario, the expected performance should be equal to that of random UAV placement. Indeed, in Fig. \ref{distalgofig1}, we can observe that the performance for the cases $D_s = D_c = 10m$ and $D_s = 10m,\,D_c = \infty$ is almost identical to that of random UAV placement, verifying our analytical intuition.

\begin{figure}
  \begin{center}
    \includegraphics[width=0.7\textwidth]{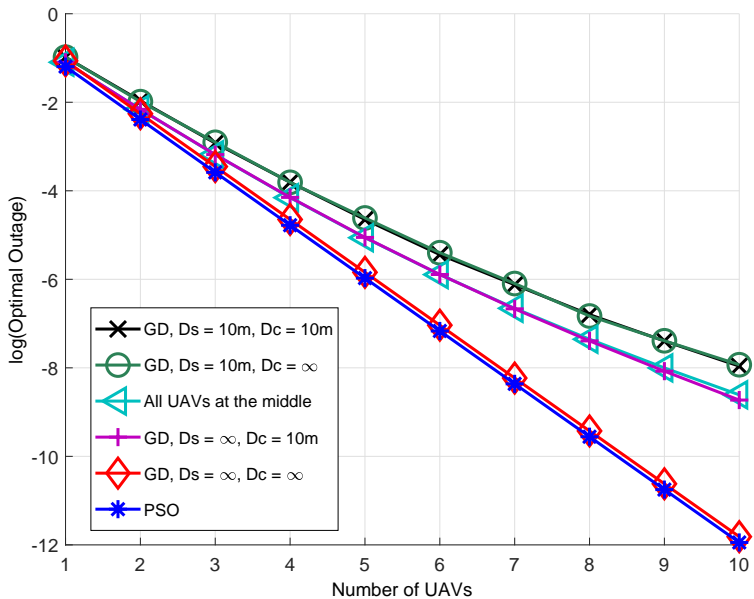}
  \end{center}
  \vspace{-10pt}
  \caption{ Comparison of UAV deployment algorithms for two-dimensional Gaussian density at $h = 300m$.}
  \label{fig7}
\end{figure}

Fig. \ref{fig7} compares the results of PSO algorithms and gradient descent with different setups for a two-dimensional Gaussian user density with zero mean and covariance matrix $100 \cdot \mathbf{I}_{2\times2}$.  We can similarly observe that as the communication and sensing distance decreases, performance becomes worse. Still, the results of Fig. \ref{fig7} show that our proposed distributed algorithms provide high performance when users are located in two dimensional space according to a unimodal density which is not necessarily uniform.  

Finally, in Fig. \ref{distalgofig5}, we compare the PSO algorithm with the modified gradient descent algorithm for the case $(D_s,D_c)=(500, 500)$ and Rayleigh fading model. We have considered a one-dimensional uniform density function $f(x) = 10^{-3}(x\in[0,1000])$, and UAVs located in different altitudes in $h\in\{100, 300, 500, 1000\}$. We can observe that, as the altitude increases, the performance with gradient descent becomes closer to the performance with PSO which provides the optimal result. Although we reduced the communication and sensing ranges considerably, which reduces the computational and energy cost, the modified gradient descent method still provides the almost optimal result for altitudes greater than some threshold value (e.g. $h=300$m in this example). This indicates that we can use our cost-efficient modified gradient descent algorithm to obtain the optimal results when the UAV altitudes are sufficiently high.

\begin{figure}
  \begin{center}
  \vspace{-20pt}
    \includegraphics[width=0.7\textwidth]{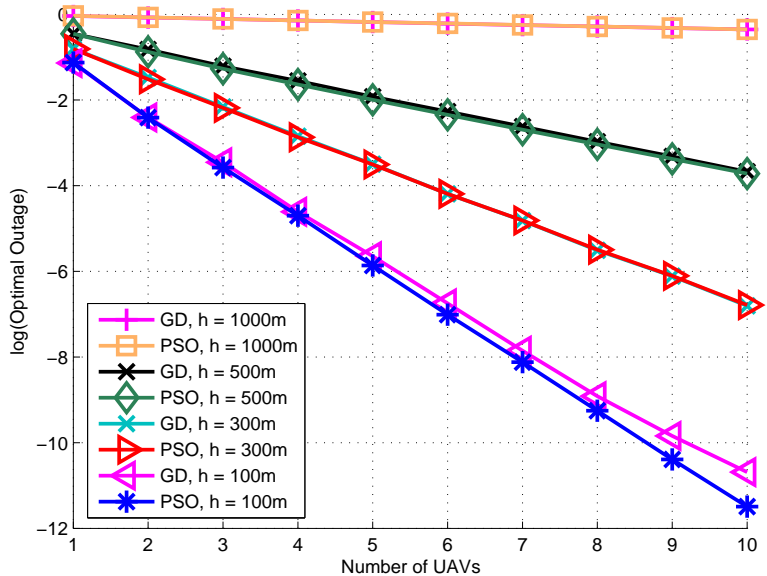}
  \end{center}
  \vspace{-10pt}
  \caption{ Algorithm comparisons for a one-dimensional uniform density and $(D_s,D_c)=(500m, 500m)$ at different altitudes.}
  \label{distalgofig5}
\end{figure}

Modifying $D_s$ and/or $D_c$ depending on the UAV altitude constraint can be considered as another dimension of adaptation in UAV-based communication systems. For example, one can determine the smallest values for $D_s$ and $D_c$ that provides an acceptable performance at each altitude. One can then adaptively change these parameters as the altitude constraints change to conserve sensing or communication power.

\hspace{20pt}
\section{Conclusions}
\label{secconclusions}
We have studied optimal placement of UAVs serving as mobile base stations to several GTs. Our objective has been to minimize the outage probability of the system. We have shown that the optimal outage probability decays exponentially with the number of UAVs. We have also proved that in an optimal deployment, all UAVs collapse to a unique location at large altitudes. We have also designed centralized and distributed algorithms to numerically optimize the UAV locations. We have also extended our results to a practical variant of  Rician fading that is specifically tailored for UAV communications. As future work, we aim to provide better bounds for general GT densities and UAV altitudes. We also plan to study the movement of UAVs according to a time-varying GT density. Moreover, our model has only utilized selection diversity for reception of GT data. It is possible to apply maximum ratio combining for a better performance. Optimal placement of UAVs for this scenario is an interesting direction for future work.

\appendices
\section{Proof of Theorem \ref{oneuavtheorem}}
\label{oneuavtheoremproof}
We need to show that a global maximum of 
\begin{align}
h(u) \triangleq \int g(x,u) f(x) \mathrm{d}x
\end{align}
is located at $u = \mu$. We first provide a proof for the case $d=1$. We need the following properties of a unimodal function. If $f$ is a univariate unimodal density with center $\mu$, we have
\begin{align}
\label{unimod1} f(\mu+c) & = f(\mu-c),\,\forall c\in\mathbb{R}, \\
\label{unimod2}  f(\mu+c) & \geq f(\mu+d) ,\, d \geq c \geq 0, \\
\label{unimod3}  f(\mu-c) & \geq f(\mu-d) ,\, d \geq c \geq 0.
\end{align} 
These properties follow immediately from Definition \ref{unimoddef}.

We are now ready to prove the theorem for $d=1$. We will show that for any $\alpha \in \mathbb{R}$, we have $h(\mu) \geq h(\mu + \alpha)$. Equivalently, defining $H \triangleq h(\mu) - h(\mu + \alpha)$, we will show that $H \geq 0$. First, we consider $\alpha \geq 0$. We have
\begin{align}
H = \int (g(x,\mu) - g(x,\mu+ \alpha))f(x)\mathrm{d}x
\end{align}
Applying a change of variables $t = x - \mu - \frac{\alpha}{2}$, we obtain
\begin{align}
H = \int \zeta(t) f(t + \mu + \tfrac{\alpha}{2})\mathrm{d}t,
\end{align}
where $\zeta(t) \triangleq g(t,-\tfrac{\alpha}{2}) - g(t,\tfrac{\alpha}{2})$. It can easily be verified that $\zeta(t) = -\zeta(-t)$ so that
\begin{align}
H &  = \int_{-\infty}^0  \zeta(t)f(t+\mu+\tfrac{\alpha}{2})\mathrm{d}t +   \int_{0}^{\infty} \zeta(t) f(t+\mu+\tfrac{\alpha}{2})\mathrm{d}t  \\
\label{laksjdlaksjdas} & =  \int_{-\infty}^0 \zeta(t) (f(t+\mu+\tfrac{\alpha}{2}) -f(-t+\mu+\tfrac{\alpha}{2})) \mathrm{d}t.
\end{align}
The equality follows from a change of variables $t \leftarrow -t$, and the fact that $\zeta(\cdot)$ is an odd function. 

Let us now partition the integration domain in (\ref{laksjdlaksjdas}) to two subsets, namely, $(-\infty,-\tfrac{\alpha}{2})$ and $[-\tfrac{\alpha}{2},0]$, and call the resulting integrals $H_1$ and $H_2$, respectively. We have
\begin{align}
H_1 & =  \int_{-\infty}^{-\tfrac{\alpha}{2}} \zeta(t) (f(t+\mu+\tfrac{\alpha}{2}) -f(-t+\mu+\tfrac{\alpha}{2})) \mathrm{d}t \\
\label{honebounddd} & =  \int_{-\infty}^{-\tfrac{\alpha}{2}} \zeta(t) (f(-t+\mu-\tfrac{\alpha}{2}) -f(-t+\mu+\tfrac{\alpha}{2})) \mathrm{d}t.
\end{align}
The equality follows from the symmetry  (\ref{unimod1}) of $f$ around $\mu$. Now, given $t \leq -\frac{\alpha}{2}$, we have $0 \leq -\frac{\alpha}{2}-t \leq \frac{\alpha}{2} - t$. According to (\ref{unimod2}), we obtain 
\begin{align}
\label{qwjepqowjepqowje1}
f(-t+\mu-\tfrac{\alpha}{2}) \geq f(-t+\mu+\tfrac{\alpha}{2}),\,t \leq -\frac{\alpha}{2} 
\end{align}
It can also be easily verified that
\begin{align}
\label{qwjepqowjepqowje2}
\zeta(t) \geq 0,\,\forall t \leq 0.
\end{align}
Applying the bounds (\ref{qwjepqowjepqowje1}) and (\ref{qwjepqowjepqowje2}) to (\ref{honebounddd}), we obtain $H_1 \geq 0$. Similarly, we can show the non-negativity of
\begin{align}
\label{dpoqwpeoiqweq}
H_2 & =  \int_{-\tfrac{\alpha}{2}}^0 \zeta(t) (f(t+\mu+\tfrac{\alpha}{2}) -f(-t+\mu+\tfrac{\alpha}{2})) \mathrm{d}t.
\end{align}
In fact, given $t\in[-\tfrac{\alpha}{2},0]$, we have the chain of inequalities $0 \leq t+\frac{\alpha}{2} \leq -t + \frac{\alpha}{2}$. By (\ref{unimod2}), we obtain $f(\mu + t+\frac{\alpha}{2}) \geq f(\mu -t + \frac{\alpha}{2})$. Applying this bound to (\ref{dpoqwpeoiqweq}) together with (\ref{qwjepqowjepqowje2}), we obtain $H_2 \geq 0$. Since $H_1 \geq 0$ as already shown, we have $H  = H_1 + H_2 \geq 0$. Using the same arguments, one can also show $H \geq 0$ for $\alpha \leq 0$. This concludes the proof for $d=1$. 

Now, suppose $d=2$, $\mu = [\begin{smallmatrix} \mu_1 \\ \mu_2 \end{smallmatrix}]$. To prove the theorem, it is sufficient to show that (i) for every $x$ and $\alpha \geq 0$, we have $h([\begin{smallmatrix} x \\ \mu_2 \end{smallmatrix}]) \geq h([\begin{smallmatrix} x \\ \mu_2+\alpha \end{smallmatrix}])$, and (ii) for every $y$ and $\alpha \geq 0$, we have $h([\begin{smallmatrix} \mu_1 \\ y \end{smallmatrix}]) \geq h([\begin{smallmatrix} \mu_1 + \alpha \\ y \end{smallmatrix}])$. Both claims can be verified using the same arguments that we have used for $d=1$. This concludes the proof of the theorem.

\section{Proof of Theorem \ref{triviallbh0}}
\label{proofoftriviallbh0}
We first consider the case $d=1$ and a uniform distribution on $[0,M]$. Without loss of generality, assume $0 \leq u_1 \leq u_2 \leq \cdots \leq u_n \leq M$. Given $j\in\{1,\ldots,n+1\}$, let $I_j \triangleq [u_{j-1},u_j]$, with the convention that $u_0 = 0$ and $u_{n+1} = 1$. We have
\begin{align}
P_{O}(U) = \sum_{j=1}^n \int_{I_j} \prod_{i=1}^n (1 - \exp(-\lambda (x-u_i)^r))\frac{\mathrm{d}x}{M}
\end{align}
For the $j$th term of the sum, instead of integrating over all $I_j = [u_{j-1},u_j]$, we integrate over 
\begin{align}
I_j' \triangleq \left[u_{j-1} + \frac{|I_j|}{3}, u_{j} - \frac{|I_j|}{3}\right],\,j\in\{1,\ldots,n\}.
\end{align}
By construction, for all $j\in\{1,\ldots,n\}$ and every $x\in I_j'$, we have $|x-u_i| \leq \frac{1}{3}|I_j|,\,\forall i\in\{1,\ldots,n\}$. It follows that
\begin{align}
\label{qojwepqowjepqoeq}
P_{O}(U) & \geq \sum_{j=1}^n \int_{I_j'} \prod_{i=1}^n (1 - \exp(-\lambda(x-u_i)^r))\frac{\mathrm{d}x}{M} \\
& \geq \sum_{j=1}^n \int_{I_j'} \prod_{i=1}^n \bigl(1 - \exp\left(-\lambda |I_j|^{r}/3^r \,\right)\bigr) \frac{\mathrm{d}x}{M} \\
&  = \sum_{j=1}^n \frac{|I_j'|}{M} \bigl(1 - \exp\left(-\lambda |I_j|^{r}/3^r \,\right)\bigr)^n\\
\label{twosumzzzz} &  = \sum_{j=1}^n \frac{|I_j|}{3M} \bigl(1 - \exp\left(-\lambda |I_j|^{r}/3^r \,\right)\bigr)^n
\end{align}
The second derivative of the function $g(x) \triangleq 1-\exp(- \lambda x^r/3^r)$ can be calculated to be  
\begin{align}
g''(x) =  \lambda r9^{-r} \exp(-\lambda x^r/3^r)x^{r-2} (3^r(r-1)-\lambda rx^r), 
\end{align}
which means that $g$ is convex whenever $0 \leq x \leq 3(\lambda)^{-1/r}(1-\frac{1}{r})^{1/r}$. By the second derivative test, it is straightforward to show that the product of two non-decreasing convex functions is convex. It follows that, with regards to the second sum in (\ref{twosumzzzz}), the function $x \mapsto x (1 - \exp(- \lambda x^r/3^r))^n $ is convex whenever $x \leq 3(\lambda)^{-1/r}(1-\frac{1}{r})^{1/r}$. 

Now, suppose there exists $ \ell \in\{1,\ldots,n\}$ such that $|I_{\ell}| \geq 3(\lambda)^{-1/r}(1-\frac{1}{r})^{1/r}$. We then have
\begin{align}
 P_{O}(U) & \geq \frac{|I_{\ell} |}{3M} \bigl(1 - \exp\left(- |I_{\ell}|^{r}/3^r \,\right)\bigr)^n \\
  & \geq \frac{(1-\frac{1}{r})^{1/r}}{M} \bigl(1 - \exp\left(-1 + \tfrac{1}{r} \,\right)\bigr)^{n},
\end{align}
and Theorem \ref{triviallbh0} follows. We can therefore assume  $|I_j| \leq 3(1-\frac{1}{r})^{1/r},\,\forall j\in\{1,\ldots,n\}$.  We now recall Jensen's inequality: For an arbitrary convex function $h$, and arbitrary real numbers $a_1,\ldots,a_n$, we have
$\frac{1}{n} \sum_{i=1}^n h(a_i) \geq h( \frac{1}{n}\sum_{i=1}^n a_i)$. Applying to (\ref{twosumzzzz}), we obtain
\begin{align}
 P_{O}(U) \geq  
 \frac{1}{3} \bigl(1 - \exp\left(- 1/3^r \,\right)\bigr)^n.
\end{align} 
This concludes the proof for $d=1$ and a uniform distribution. The case of a uniform distribution on $[0,M]^2$ in $d=2$ dimensions can be handled similarly by choosing $I_j$s to be a sequence of disks that cover $[0,M]^2$. Correspondingly, for each $j$, we choose $I_j'$ to be the disk that is concentric to $I_j$, but with one-third the radius. The case of a non-uniform density can be handled by considering a box where the density is bounded from below by a positive constant and applying the results already found for a uniform density. This concludes the proof of the theorem in the general case.

\section{Proof of Theorem \ref{asymptottheorem}}
\label{proofofasymptottheorem}
We have the expansion
\begin{align}
P_O(U)  & =  \int \prod_{i=1}^n (1- g(x, u_i)) f(x)\mathrm{d}x
 \\ \label{dlaksjdlaksda} & =  \int \biggl[1 - \sum_{i=1}^n g(x,u_i)  + \sum_{k=2}^n (-1)^k  \sum_{J \in C_{n,k}} \prod_{j \in J} g(x,u_j) \biggr] f(x)\mathrm{d}x,
\end{align}
where $C_{n,k}$ is the collection of all $k$-combinations of the set $\{1,\ldots,n\}$ (For example, $C_{3,2} = \{\{1,2\},\{1,3\},\{2,3\}\}$). Now, note that 
$g(x,u) = \exp\left( -(\|x-u\|^2 + h^2)^{\frac{r}{2}} \right) \leq e^{-h^r}$. Therefore, for any $J\in C_{n.k}$, we have 
\begin{align}
\label{jlkdasjdlkasjd}
\int  \prod_{j \in J} g(x,u_j) f(x)\mathrm{d}x \leq e^{-|J| h^r}.
\end{align}
Using (\ref{jlkdasjdlkasjd}) and the bound $(-1)^k \leq 1$ in (\ref{dlaksjdlaksda}), we obtain
\begin{align}
P_O(U)  & \leq  1 - \sum_{i=1}^n \int g(x,u_i) f(x)\mathrm{d}x + \sum_{k=2}^n \binom{n}{k} e^{-kh^r} \\
 & \leq  1 - \sum_{i=1}^n \int g(x,u_i) f(x)\mathrm{d}x + 2^n e^{-2h^r}
\end{align}
Taking the minimum over all deployments, we have
\begin{align}
\label{alskdjalksdjazz}
P_O(U^{\star})  \leq 1 - n \int g(x,u^{\star}(h)) f(x)\mathrm{d}x +2^n e^{-2h^r}.
\end{align}
A similar argument results in the converse estimate
\begin{align}
\label{alskdjalksdjazzconv}
P_O(U^{\star})  \geq  1 - n \int g(x,u^{\star}(h)) f(x)\mathrm{d}x -2^n e^{-2h^r}.
\end{align}
The statement of the theorem follows from (\ref{alskdjalksdjazz}), (\ref{alskdjalksdjazzconv}), and the dominated convergence theorem. 

\section{Proof of Theorem \ref{oneuavtheoremRician}}
\label{appendixRician}
In the following, we show that a global maximum of 
\begin{align}
h(u) \triangleq \int w(x,u) f(x) \mathrm{d}x
\end{align}
is located at $u = \mu$. We will show that for any $\alpha \in \mathbb{R}$, we have $h(\mu) \geq h(\mu + \alpha)$. Equivalently, defining $H \triangleq h(\mu) - h(\mu + \alpha)$, we will show that $H \geq 0$. First, we consider $\alpha \geq 0$. We have
\begin{align}
H = \int (w(x,\mu) - w(x,\mu+ \alpha))f(x)\mathrm{d}x
\end{align}
Applying a change of variables $t = x - \mu - \frac{\alpha}{2}$, we obtain 
\begin{align}
H = \int \zeta(t) f(t + \mu + \tfrac{\alpha}{2})\mathrm{d}t,
\end{align}
where $\zeta(t) \triangleq w(t,-\tfrac{\alpha}{2}) - w(t,\tfrac{\alpha}{2})$. We now have
\begin{align}
\zeta(t) & \nonumber = \mathrm{Q}\left(\sqrt{2K \tan^{-1}\tfrac{h}{\|t+\frac{\alpha}{2}\|}},\sqrt{2\lambda \bigl(K\tan^{-1}\tfrac{h}{\|t-\tfrac{\alpha}{2}\|}+1\bigr)(\|t+\tfrac{\alpha}{2}\|^2 + h^2)^{r\tan^{-1}\frac{h}{\|t+\frac{\alpha}{2}\|}}} \right) \\
& - \mathrm{Q}\left(\sqrt{2K\tan^{-1}\tfrac{h}{\|t-\frac{\alpha}{2}\|}}, \sqrt{2\lambda \bigl(K\tan^{-1}\tfrac{h}{\|t-\tfrac{\alpha}{2}\|}+1\bigr)(\|t-\tfrac{\alpha}{2}\|^2 + h^2)^{r\tan^{-1}\frac{h}{\|t-\frac{\alpha}{2}\|}}}\right) \\
 & = - [ w(-t,-\tfrac{\alpha}{2}) + w(-t,\tfrac{\alpha}{2})] = -\zeta(-t),
\end{align}
so that $\zeta(\cdot)$ is an odd function. The proof can now be completed using the same arguments as in the proof of Theorem \ref{oneuavtheorem}.

\section{Proof of Theorem \ref{largeHtheoremRician}}
\label{appendixRician2}
We utilize the alternative form of the Marcum Q-function
\begin{equation}
\label{QdefSeries}
Q_1(a,b) = \exp{\left(-\frac{a^2+b^2}{2}\right)} \sum_{k=0}^{\infty} \left(\frac{a}{b}\right)^k I_k(ab).
\end{equation}
Also, from \cite{415987} we have
\begin{equation}
\label{1178}
I_{k}(t) \leq \frac{t^k e^t}{2^k k!}   
\end{equation}
To prove (\ref{1178}), first note that \cite{walsh1972}
\begin{equation}
\label{1177}
\Gamma (k + 1) \left( \frac{2}{t} \right)^{k} I_{k}(t) < \cosh (t)
\end{equation}
for $t>0$ and $k > -1/2$. In addition, for $t>0$ we have $\cosh (t) - e^{t} = -\sinh(t) <0$, and thus $\cosh (t) < e^{t}$. Substituting to (\ref{1177}), we obtain
\begin{equation}
\Gamma (k + 1) \left( \frac{2}{t} \right)^{k} I_{k}(t) < e^{t},
\end{equation}
which yields (\ref{1178}). Hence, we have
\begin{align}
Q_1(a,b) \leq \exp \left(-\frac{a^2+b^2}{2}\right) \sum_{k=0}^{\infty} \left(\frac{a}{b}\right)^k \frac{(ab)^k e^{ab}}{2^k k!}  \\
\label{1122}
= \exp{\left(-\frac{a^2+b^2}{2}+ab\right)} \sum_{k=0}^{\infty} \left(\frac{a^2}{2}\right)^k \frac{1}{k!} 
\end{align}
Also, using $\sum_{k=0}^{\infty} \frac{t^k}{k!} = e^{t}$ and (\ref{1122}), we obtain
\begin{align}
\label{Qbound}
Q_1(a,b) \leq  \exp{\left(-\frac{a^2+b^2}{2}+ab + \frac{a^2}{2}\right)} =  \exp{\left(-\frac{b^2}{2}+ab \right)}.
\end{align}

We can now expand the outage probability as
\begin{align}
P_O(U)  & =  \int \prod_{i=1}^n (1- w(x, u_i)) f(x)\mathrm{d}x
 \\ \label{7373} & =  \int \biggl[1 - \sum_{i=1}^n w(x,u_i)  + \sum_{k=2}^n (-1)^k  \sum_{J \in C_{n,k}} \prod_{j \in J} w(x,u_j) \biggr] f(x)\mathrm{d}x.
\end{align}


In particular, for the terms $w(x,u_i)$, we apply (\ref{Qbound}) to the definition of $w(x,u_i)$ to obtain
\begin{multline}
\label{1144}
w(x,u_i) \leq  \exp ( -\lambda (K(\theta_i)+1)(\|x-u_i\|^2 + h^2)^{r(\theta_i)/2} +\\
\sqrt{4\lambda K(\theta_i) (K(\theta_i)+1)}(\|x-u_i\|^2 + h^2)^{r(\theta_i)/4}). 
\end{multline}
Furthermore, as $h\rightarrow\infty$, we have $\theta \rightarrow \frac{\pi}{2}$. Correspondingly, let $K' \triangleq K(\frac{\pi}{2}) = a_{3}.e^{b_3 \frac{\pi}{2}}$, $P' \triangleq  P_{LOS}(\frac{\pi}{2}) = \frac{1}{1+a_2 e^{-b_2(\frac{\pi}{2}-a_2)}}$, and $r' \triangleq r(\frac{\pi}{2}) = a_1P_{{\infty}}+b_1$ denote the limiting value for various parameters for large elevations. Also, let $K_1= \lambda(K_{\infty}+1)$, $K_2 = \sqrt{4 \lambda \mathrm{K}_{\infty} (\mathrm{K}_{\infty}+1)}$, and $\tau$ be the maximum distance between any two points in the support of the user density. For large elevations, (\ref{1144}) can be bounded as
\begin{equation}
\label{1167}
w(x, u_i) \leq  e^{ -K_1 h^{r'} + K_2(\tau^2 + h^2)^{r'/4} }.
\end{equation}
The upper bound goes to $0$ as $K_1>0$ and  $(\tau^2 + h^2)^{r'/4} = o(h^{r'})$ as $h\to \infty$. Therefore, for any $J\in C_{n,k}$, we have 
\begin{align}
\label{jjdsjd}
\int  \prod_{j \in J} w(x,u_j) f(x)\mathrm{d}x \leq
 e^{ -K_1 |J| h^{r'} + K_2|J|(\tau^2 + h^2)^{r'/4} } 
\end{align}

Using (\ref{jjdsjd}) and the bound $(-1)^k \leq 1$ in (\ref{7373}), we obtain
\begin{align}
\label{lskdlsdksldks}
P_O(U)  & \leq  1 - \sum_{i=1}^n \int w(x,u_i) f(x)\mathrm{d}x + \sum_{j=2}^n \binom{n}{j} e^{ -K_1 j h^{r'} + K_2 j(\tau^2 + h^2)^{r'/4} }.
\end{align}

Since $e^{ -K_1 j h^{r'} + K_2 j(\tau^2 + h^2)^{r'/4} } $ is a decreasing function of $j$, we obtain
\begin{align}
\label{lemmamain}
\sum_{j=2}^n \binom{n}{j} e^{-jK_1 h^{r'}} I_0^{j}(K_2 (\tau^2 + h^2)^{\frac{r'}{4} } \leq 2^{n}  e^{ -2K_1 h^{r'} + 2K_2(\tau^2 + h^2)^{r'/4} } 
 \end{align}
Substituting to (\ref{lskdlsdksldks}) and taking the minimum over all deployments, we have
\begin{align}
\label{alskdjalksdjazz}
P_O(U^{\star})  \leq 1 - n \int w(x,u^{\star}(h)) f(x)\mathrm{d}x + 2^{n}  e^{ -2K_1 h^{r'} + 2K_2(\tau^2 + h^2)^{r'/4} } .
\end{align}
A similar argument results in the converse estimate
\begin{align}
\label{alskdjalksdjazzconv}
P_O(U^{\star})  \geq  1 - n \int w(x,u^{\star}(h)) f(x)\mathrm{d}x - 2^{n}  e^{ -2K_1 h^{r'} + 2K_2(\tau^2 + h^2)^{r'/4} } .
\end{align}
The statement of the theorem follows from (\ref{alskdjalksdjazz}), (\ref{alskdjalksdjazzconv}), and the dominated convergence theorem.

\linespread{1.8}
\bibliography{main} 
\bibliographystyle{IEEEtran}

\end{document}